\documentclass[draftcls, 12pt, onecolumn]{IEEEtran}

\IEEEoverridecommandlockouts                              % This command is only                                       % needed if you want t                                   % use the \thanks command
\usepackage[normalem]{ulem }
\usepackage{draftcopy}
\usepackage{color}
\usepackage{graphicx}
\usepackage{epsfig}
\usepackage{amsmath}
\usepackage{amssymb}
\usepackage{latexsym}
\usepackage{amsfonts}
\usepackage{ulem}
\usepackage{cite}
\usepackage{subfigure}
\usepackage{graphics}
\usepackage{latexsym}
\usepackage{graphicx}
\usepackage{makeidx}
\usepackage{amsmath}
\usepackage{amssymb}
\usepackage{psfrag}
\usepackage{hyperref}
 
 \newtheorem{lemma}{Lemma}
 \newtheorem{proposition}{Proposition}
  
 \newtheorem{definition}{Definition}
 \newtheorem{remark}{Remark}
 
 \newtheorem{example}{Example}

\renewcommand{\emph}{\textit}

\begin{document}
\title{Minimal realization of the dynamical structure function and its application to network reconstruction}
\author{Ye Yuan, Guy-Bart Stan, Sean Warnick and Jorge Gon\c{c}alves
  \thanks{Ye Yuan and Jorge Gon\c{c}alves are with
    Control Group, Department of Engineering, University of Cambridge (yy311@cam.ac.uk, jmg77@cam.ac.uk). Guy-Bart Stan is with the Department of Bioengineering and the Centre for Synthetic Biology and Innovation, Imperial College London (g.stan@imperial.ac.uk). Sean Warnick is with the Information and Decision Algorithms Laboratories, Department of Computer Science, Brigham Young University (sean.warnick@gmail.com). }    }
\maketitle \thispagestyle{empty}

\begin{abstract}
  Network reconstruction, i.e., obtaining network structure from data,
  is a central theme in systems biology, economics and
  engineering. In some previous work, we introduced dynamical structure functions
  as a tool for posing and solving the problem of network
  reconstruction between measured states. While recovering the network
  structure between hidden states is not possible since they are not
  measured, in many situations it is important to estimate the minimal number
  of hidden states in order to understand the complexity of the
  network under investigation and help identify potential targets for
  measurements. Estimating the minimal number of hidden states is also crucial
  to obtain the simplest state-space model that captures the network
  structure and is coherent with the measured data.  This paper
  characterizes minimal order state-space realizations that are
  consistent with a given dynamical structure function by exploring
  properties of dynamical structure functions and developing
  an algorithm to explicitly obtain such a minimal realization.
  \end{abstract}

\section{Introduction}
Networks have received an increasing amount of attention in the last
decade. In our ``information-rich'' world, the questions of network
reconstruction and network analysis become crucial for the
understanding of complex systems such as biological, social, or
economical networks. In particular, the analysis of molecular
networks has gained significant interest due to the recent explosion
of publicly available high-throughput biological data. In this context, the question of
identifying and analyzing the network structure at the
origin of measured data becomes a key issue.

In some occasions, measured data is given in the form of input-output
time-series that describes the effect of inputs on outputs (measured
states) of a network. When data is generated by a linear system, a
matrix transfer function describing the dynamic input-output behavior is generally obtained using system identification~\cite{Ljung}.
If the original state-space model is available or deducible, then the associated network structure can be readily obtained from it. However, a transfer function cannot, in general,
recover, or realize, the original state-space model since the
realization problem does not typically have a unique solution, i.e.,
different state-space realizations can generate the same input-output
behavior.
%There is a well-developed realization theory for linear systems that answers
%questions such as what is the minimal number of states needed to describe a
%given strictly proper transfer function ${G}(s)$?, or how can we find
%state space realizations of particular canonical forms? or how can we
%relate the state space matrices (${A}$, ${B}$, ${C}$) to the transfer function
%${G}(s)$.
Since each of these realizations may suggest entirely different
network structures, it is in general impossible to identify network
structures from transfer functions alone. Therefore, more information,
beyond input-output data used to identify a transfer function, is
needed to prefer one state-space realization over another as a
description of a particular system\cite{csm}.

Another difficulty in the network reconstruction problem comes from
the fact that the realization problem becomes ill posed when some of the states
are unobservable or ``hidden'' (this happens with just
one hidden state \cite[pp. $78$]{zdg}). As a result, failure to
explicitly acknowledge the presence of hidden states and the resulting
ambiguity in network structures can lead to a deceptive and erroneous
process for network structure discovery. Consequently, determining
from measured data the presence or absence of a causal relationship
between two variables in a network is a challenging question. 

Motivated by this, we are focusing on the effect of hidden states in the network that we are aiming to reconstruct.  A new representation for LTI systems, called dynamical structure functions was introduced in \cite{08net_rec} and developed in \cite{russ1, russ2, book, enoch1, enoch2, enoch3} .
Dynamical structure functions capture information at an intermediate level between transfer function and state space
representation (see Figure~\ref{Fig:math}). Specifically, dynamical structure functions not only encode structural information at the measurement level, but also contain some
information about hidden states. Based on the theoretical results presented in \cite{08net_rec}, we proposed some guidelines for the design of an experimental data-acquisition protocol
which allows the collection of data containing sufficient information for the network structure reconstruction problem
to become solvable. In particular, we have shown that if nothing is known about the network, then reconstruction is impossible.  If, however, one can make the following reasonable assumptions about the data-collection experiments:
Networks have received an increasing amount of attention in the last
decade. In our ``information-rich'' world, the questions of network
reconstruction and network analysis become crucial for the
understanding of complex systems such as biological, social, or
economical networks. In particular, the analysis of molecular
networks has gained significant interest due to the recent explosion
of publicly available high-throughput biological data. In this context, the question of
identifying and analyzing the network structure at the
origin of measured data becomes a key issue.

In some occasions, measured data is given in the form of input-output
time-series that describes the effect of inputs on outputs (measured
states) of a network. When data is generated by a linear system, a
matrix transfer function describing the dynamic input-output behavior is generally obtained using system identification~\cite{Ljung}.
If the original state-space model is available or deducible, then the associated network structure can be readily obtained from it. However, a transfer function cannot, in general,
recover, or realize, the original state-space model since the
realization problem does not typically have a unique solution, i.e.,
different state-space realizations can generate the same input-output
behavior.
%There is a well-developed realization theory for linear systems that answers
%questions such as what is the minimal number of states needed to describe a
%given strictly proper transfer function ${G}(s)$?, or how can we find
%state space realizations of particular canonical forms? or how can we
%relate the state space matrices (${A}$, ${B}$, ${C}$) to the transfer function
%${G}(s)$.
Since each of these realizations may suggest entirely different
network structures, it is in general impossible to identify network
structures from transfer functions alone. Therefore, more information,
beyond input-output data used to identify a transfer function, is
needed to prefer one state-space realization over another as a
description of a particular system\cite{csm}.

Another difficulty in the network reconstruction problem comes from
the fact that the realization problem becomes ill posed when some of the states
are unobservable or ``hidden'' (this happens with just
one hidden state \cite[pp. $78$]{zdg}). As a result, failure to
explicitly acknowledge the presence of hidden states and the resulting
ambiguity in network structures can lead to a deceptive and erroneous
process for network structure discovery. Consequently, determining
from measured data the presence or absence of a causal relationship
between two variables in a network is a challenging question. 

Motivated by this, we are focusing on the effect of hidden states in the network that we are aiming to reconstruct.  A new representation for LTI systems, called dynamical structure functions was introduced in \cite{08net_rec} and developed in \cite{russ1, russ2, book, enoch1, enoch2, enoch3} .
Dynamical structure functions capture information at an intermediate level between transfer function and state space
representation (see Figure~\ref{Fig:math}). Specifically, dynamical structure functions not only encode structural information at the measurement level, but also contain some
information about hidden states. Based on the theoretical results presented in \cite{08net_rec}, we proposed some guidelines for the design of an experimental data-acquisition protocol
which allows the collection of data containing sufficient information for the network structure reconstruction problem
to become solvable. In particular, we have shown that if nothing is known about the network, then reconstruction is impossible.  If, however, one can make the following reasonable assumptions about the data-collection experiments:
\begin{enumerate}
\item[(A.1)] for a network composed of $p$ measured species, the same
number of experiments $p$ must be performed; 
\item[(A.2)] each experiment
must independently control a measured species, i.e., control input $i$
must first affect measured species $i$,
\end{enumerate}
then reconstruction is possible.  Moreover, failure to meet the necessary informativity conditions results in a situation where any internal network structure fits the data equally well
(e.g. a fully decoupled network or a fully connected network).  If
biologists have already some information about the network, as it is
usually the case, then these conditions can be relaxed as explained
in~\cite{08net_rec}.  Using dynamical structure functions as a mean to solve the network reconstruction problem, the following aspects need to be considered (see Figure~\ref{Fig:math}):
\begin{figure}
\centering
\includegraphics[scale=0.34]{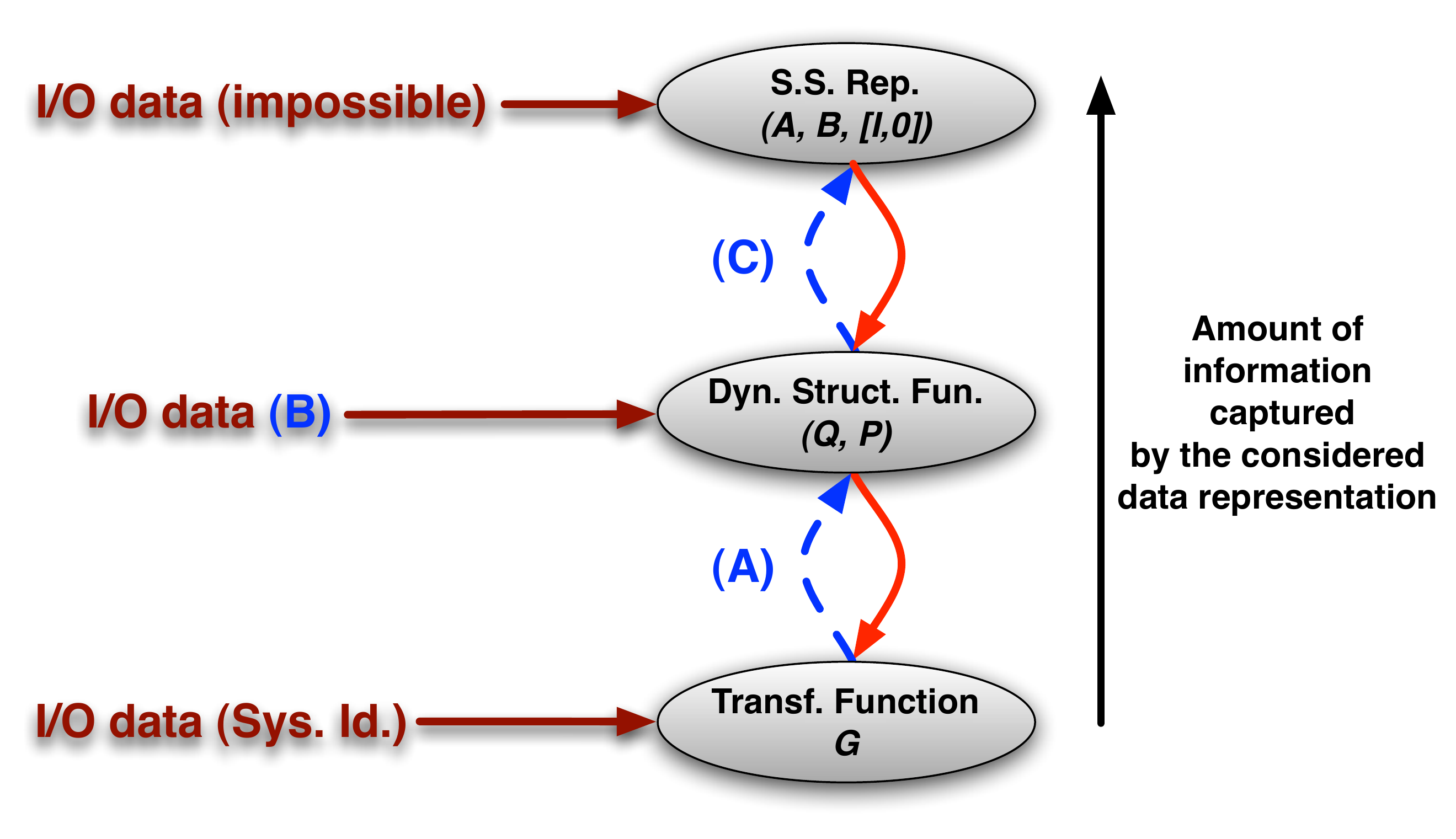}
\caption{Mathematical structure of the network reconstruction problem using dynamical structure functions. Red arrows mean ``uniquely determine'', blue arrows indicate our work.}\label{Fig:math}
\end{figure}

First (see (A) in Figure~\ref{Fig:math}), the properties of a dynamical
structure function and its relationship with the transfer
function associated with the same system need to be precisely established (this was done in \cite{08net_rec}).
%We will show that if
%experiments are performed as explained above: 1) we can not only
%obtain the network between the measured states but also the
%``self-loop'' gain for each measured state; 2) we can still recover
%the true network structure even if the exact value of control
%input is unknown. 

Second (see (B) in Figure~\ref{Fig:math}), an efficient method is developed to reconstruct networks in the presence of noise and nonlinearities (this was done in~\cite{robust}). This method relies on the assumption that the conditions for network reconstruction presented above in (A.1) and (A.2) have been met. In our approach, we use the same information as traditional system identification methods, i.e., input-output data. However, with our method, steady-state (resp. time-series data) can be used to reconstruct the Boolean (resp. dynamical network) structure of the system (see \cite{robust} for more details). 

Third (see (C) in Figure~\ref{Fig:math}), once the dynamical structure function is obtained, as a main result of this paper, an algorithm for constructing a minimal order state-space representation consistent with such function is developed. 
In an application, this provides a way to estimate the complexity of the
system by determining the minimal number of hidden states in the
system. For example, in the context of biology it helps understand
the number of unmeasured molecules in a particular pathway: a low
number means that most molecules in that pathway have been identified
and measured, showing a good understanding of the system; while a
large number shows that there are still many unmeasured variables,
suggesting that new experiments should be carried out to better
characterize that pathway.

The outline of the paper is as follows.
Section~\ref{sec:systemmodel} reviews the definition of dynamical
structure functions and their properties. The main
result can be found in Section~\ref{sec:optimalD} where we propose a minimal order realization algorithm based on state-space
realizations and pole-zero analysis. Simulation and discussion are addressed in Section~\ref{Sec:simu}. Finally conclusions are presented in
Section~\ref{sec:conclusion}.

\section{System Model}\label{sec:systemmodel}
Consider a linear system (it can also be a linearization of some original nonlinear system) 
$\dot{{x}}=
{Ax}+{Bu}$, ${y}={Cx}$.
The transfer function associated with this system is given by ${G(s)} =
{C}(s{I}-{A})^{-1}{B}$.  Typically, we can use
standard system identification tools~\cite{Ljung}
to identify a transfer function ${G(s)}$ from input-output data.

Like system realization, network reconstruction also begins with the identification of a
transfer function, but it additionally attempts to determine the network structure
between measured states without imposing any additional structure on
the hidden states. As we have shown in \cite{08net_rec}, this requires a new representation of linear time-invariant systems: the dynamical structure function (defined later). An algorithm allowing the dynamical structure function to be obtained from input-output data is proposed in \cite{robust}. This paper assumes that the dynamical structure function has already been obtained from data, and will focus on finding one of its minimal state-space realizations.

The dynamical structure function is obtained as follows: First, we transform
$[{A},{B},{C}]$ to
$\left[{A}^o,{B}^o,\begin{bmatrix}{I}_p &
{0}\end{bmatrix}\right]$ without changing $G(s)$, where $p=\text{rank}(C)$. The linear system
dynamics then writes
\begin{equation}\label{eq:LTI}
 \begin{array}{cll}
   \left[\begin{array}{c}{\dot{y}}\\{\dot{z}} \end{array}\right]& = &
   \left[\begin{array}{cc}{A}^o_{11}&{A}^o_{12}\\{A}^o_{21}&{A}^o_{22}\end{array}\right]\left[\begin{array}{c}{y}\\{z}\end{array}\right]
   +\left[\begin{array}{c} {B}^o_{1}\\{B}^o_{2}\end{array}\right]{u} \\
   {y} & = &
   \left[\begin{array}{cc}{I}_p&{0}\end{array}\right]\left[\begin{array}{c}{y}\\{z}\end{array}\right]
   \end{array}
\end{equation}
where ${x}=({y},{z}) \in \Bbb R^{n^o}$ is the full
state vector, ${y}\in \Bbb R^p$ is a partial measurement of the
state, ${z}$ are the $n^o-p$ ``hidden'' states, and
${u}\in \Bbb R^m$ is the control input. In this work we
restrict our attention to situations where output measurements
constitute partial state information, i.e., $p< n^o$.  We consider
only systems with full rank transfer functions that do not have
entire rows or columns of zeros, since such ``disconnected'' systems
are somewhat pathological and only serve to complicate the
exposition without fundamentally altering our conclusions.

Taking the Laplace transforms of the signals in~(\ref{eq:LTI}) yields
\begin{equation}\label{eq:LTIlaplace}
 \begin{array}{lll}
   \left[\begin{array}{c}s{Y}\\s{Z} \end{array}\right]& = & \left[\begin{array}{cc}{A}^o_{11}&{A}^o_{12}\\{A}^o_{21}&{A}^o_{22}\end{array}\right]\left[\begin{array}{c}{Y}\\{Z}\end{array}\right] +\left[\begin{array}{c} {B}^o_{1}\\{B}^o_{2}\end{array}\right]{U}
     \end{array}
\end{equation}
where ${Y}$, ${Z}$, and ${U}$ are the Laplace
transforms of ${y}$, ${z}$, and ${u}$, respectively.
Solving for ${Z}$ gives
$${Z}=\left ( s{I} - {A}^o_{22} \right )^{-1}
{A}^o_{21} {Y} + \left ( s{I} - {A}^o_{22}
\right )^{-1} {B}^o_{2} {U}$$
Substituting this last expression of ${Z}$ into~(\ref{eq:LTIlaplace}) then yields
\begin{equation}
\label{eq:WV} s{Y} = {W}^o {Y} + {V}^o{U}
\end{equation}
where ${W}^o={A}^o_{11} + {A}^o_{12}\left ( s{I} -
  {A}^o_{22} \right )^{-1} {A}^o_{21}$ and
${V}^o={B}^o_{1} +{A}^o_{12}\left ( s{I} - {A}^o_{22} \right
)^{-1} {B}^o_{2}$.

Now, let ${R}^o$ be a
diagonal matrix formed of the diagonal terms of ${W}^o$ on its
diagonal, i.e., ${R}^o=\mbox{diag}\{{W}^o\} =
\mbox{diag}(W^o_{11}, W^o_{22}, ..., W^o_{pp})$. Subtracting ${R}^o{Y}$ from both sides of \eqref{eq:WV}, we obtain:
$$\left ( s{I} - {R}^o \right ) {Y} = \left (
  {W}^o-{R}^o \right ) {Y} + {V}^o {U}$$
Note that ${W}^o-{R}^o$ is a matrix with zeros on its
diagonal. We thus have:
\begin{equation}
  \label{eq:PQ} {Y} =  {QY} + {PU}
\end{equation}
where
\begin{equation}\label{eq:Q}
  {Q} = \left ( s{I} - {R}^o \right )^{-1} \left (
    {W}^o-{R}^o \right )
\end{equation}
and
\begin{equation}\label{eq:P}
  {P}=\left ( s{I}- {R}^o \right )^{-1} {V}^o
\end{equation}
Note that ${Q}$ is zero on the diagonal.

\begin{definition}
  Given the system~(\ref{eq:LTI}), we define the {\it dynamical
    structure function} of the system to be $[{Q},{P}]$.
%  where ${Q}$ and ${P}$ are the \textit{internal structure}
%  and \textit{control structure}, respectively, as defined
%  in~(\ref{eq:Q}) and~(\ref{eq:P}).
\end{definition}

Note that, in general, ${Q(s)}$ and ${P(s)}$ carry a lot
more information than ${G(s)}$. This can be seen from the equality
${G}(s) = ({I}-{Q}(s))^{-1}{P}(s)$ (see~\cite{08net_rec}
for details).  However, ${Q(s)}$ and ${P(s)}$ carry less
information than the state-space model \eqref{eq:LTI} (see~\cite{robust}).

\begin{definition} A dynamical structure function,
  $[{Q},{P}]$, is said to be \textit{consistent} with a
  particular transfer function, ${G}$, if there exists a
  realization of ${G}$, of some order, and of the
  form~(\ref{eq:LTI}), such that $[{Q},{P}]$ are specified
  by~(\ref{eq:Q}) and~(\ref{eq:P}). Likewise, a realization is
  consistent with $[{Q},{P}]$ if that realization gives
  $[{Q},{P}]$ from~(\ref{eq:Q}) and~(\ref{eq:P}).
\end{definition}

%\begin{definition}
%  Consider a system characterized by a transfer function ${G}$.
%  The dynamical structure of the system can be \textit{reconstructed},
%  if there is only one admissible dynamical structure function,
%  $[{Q},{P}]$, that is consistent with ${G}$.  A
%  realization of the dynamical structure function is defined as
%  \textit{reconstruction}. Likewise, the Boolean structure of the
%  system can be reconstructed if all admissible dynamical structure
%  functions that are consistent with ${G}$ have the same Boolean
%  structure.
%\end{definition}
%
%Given only a transfer function ${G}$, \cite{08net_rec} shows that
%dynamical structure reconstruction is not possible.  More information
%is required, i.e., dynamical structure reconstruction is possible from
%${G}$ if and only if in addition $p-1$ elements in each column of
%$[{Q}\ \ {P}]^T$ are known that uniquely specify the
%component of $[{Q},{P}]$ in the nullspace of $[{G}^T \
%\ {I}]$ (see \cite{08net_rec} for more details).

\begin{definition}
  We say that a realization is ${G}$ minimal if this realization
  corresponds to a minimal realization of ${G}$. We say that a
  realization is $[{Q},{P}]$ minimal if this realization is consistent with $[{Q},{P}]$ and its order is smaller than or equal to that of all realizations consistent with
  $[{Q},{P}]$.
\end{definition}

The underlying principle to find a $[{Q},{P}]$ minimal
realization is to search for a realization with the minimal number of
hidden states.  Such a realization is characterized by the minimal
number of pole-zero cancellations in the transfer functions ${Q}$ and
${P}$.

\begin{proposition}\label{th:qp}
Given a dynamical system~(\ref{eq:LTI}) and the associated dynamical structure
  functions $[{Q},{P}]$ with ${R}^o$ constructed as
  explained above (see \eqref{eq:LTI}-\eqref{eq:P}), the following conditions must hold
\begin{align}
  \text{diag}\{{A}^o_{11}\}& = \lim_{s\rightarrow\infty}{R}^o(s)\label{eq:Ds};\\
  {A}^o_{11}-\text{diag}\{{A}^o_{11}\}&=\lim_{s\rightarrow \infty} s {Q}(s)\label{eq:Qs};\\
  {B}^o_1&=\lim_{s\rightarrow \infty} s {P}(s)\label{eq:Ps}.
\end{align}
\end{proposition}
\begin{IEEEproof}
Eq.~\eqref{eq:Ds} is directly obtained from the definition of ${R}^o(s)$:
\begin{align*}
  \lim_{s\rightarrow\infty}{R}^o(s)&=\lim_{s\rightarrow\infty}\text{diag}\{{W}^o(s)\}\\
  &=\text{diag}\{\lim_{s\rightarrow\infty}{W}^o(s)\} =
  \text{diag}\{{A}_{11}^o\}
\end{align*}
Since the proofs for eq.~\eqref{eq:Qs} and~\eqref{eq:Ps} are very
similar, we focus on eq.~\eqref{eq:Qs} only.  Using the fact that
for any square matrix ${M}$, if $M^n \rightarrow 0$ when $n\rightarrow +\infty$, then  
$({I}-{M})^{-1}=\sum_{i=0}^{\infty}{M}^i$, we obtain, from
the definition of ${Q}$ given in \eqref{eq:Q},
${Q}(s)=\sum_{i=1}^{\infty}
s^{-i}{R}^{o~i-1}(s)\left({W}^o(s)-{R}^o(s)\right)$ and
${W}^o(s)={A}^o_{11}+\sum_{i=1}^{\infty}s^{-i}{A}^o_{12}{A}_{22}^{o~i-1}{A}^o_{21}$, when $s\rightarrow+\infty$.
Hence,
${Q}(s)=({A}^o_{11}-{R}^o(s))s^{-1}+{r}(s)$, in
which ${r}(s)$ is a matrix polynomial of $s$, whose largest degree is
$-2$. Finally, multiplying by $s$ on both sides and taking the limit
as $s$ goes to $\infty$ results in eq.~\eqref{eq:Qs}.  A similar
argument can be used to prove eq.~\eqref{eq:Ps}.
\end{IEEEproof}

\begin{remark}
This Proposition reveals an important property of dynamical structure functions: they encode the direct causal relationships between observed variables. 
\end{remark}
%
%\begin{proposition}\label{lemma:tra}\cite{robust}
%Let $(A,B,\begin{bmatrix}
%I & 0 \end{bmatrix})$ be a realization consistent to $[Q,P]$, then
%a linear transformation mapping from 
%$(A,B,\begin{bmatrix}
%I & 0 \end{bmatrix})$ to $(T^{-1}AT,T^{-1}B,\begin{bmatrix}
%I & 0 \end{bmatrix}T)$ by selecting $$T=\begin{bmatrix}
%I & 0 \\
%0 & T_2^{-1}
%\end{bmatrix}.$$ Such linear transformations do not change 
%$(Q,P)$.
%\end{proposition} 

We present hereafter an illustrative example to help fix the ideas. 
\begin{example} \label{ex1}
Consider a network with the structure depicted in Fig. \ref{fig:ex1}. A linear state-space representation of this network is given by
\begin{equation*} \begin{array}{rcl}
\dot{x} &=& 
\begin{bmatrix}a_{11} & 0 & a_{13} & 0 & 0 \\
			0 & a_{22} & 0 & a_{24} & 0 \\
			0 & a_{32} & a_{33} & 0 & a_{35} \\
			a_{41} & 0 & 0 & a_{44} & 0 \\
			0 & a_{52} & 0 & 0 & a_{55} \end{bmatrix} x + 
\begin{bmatrix}b_{11} & 0 \\ 0 & b_{22} \\ 0 & 0 \\ 0 & 0 \\ 0 & 0 \end{bmatrix} u \\
y &=& \begin{bmatrix} I_3 & 0 \end{bmatrix} x
\end{array} \end{equation*}
where $I_3$ is the $3\times3$ identity matrix. Following the definitions in \eqref{eq:Q} and \eqref{eq:P}, we can write down the corresponding dynamical structure function $[{Q},{P}]$ as
\begin{align*}
Q &= \begin{pmatrix} 0 & 0 & \frac{a_{13}}{s-a_{11}} \\ \frac{a_{24}a_{41}}{(s-a_{22})(s-a_{44})} & 0 & 0 \\ 0 & \frac{a_{35}a_{52} + a_{32}(s-a_{55})}{(s-a_{33})(s-a_{55})} & 0 \end{pmatrix}, \\
P &= \begin{pmatrix} \frac{b_{11}}{s-a_{11}} & 0 \\ 0 & \frac{b_{22}}{s-a_{22}} \\ 0 & 0 \end{pmatrix}.
\end{align*}

\begin{figure}[hb] \centering
\includegraphics[width=0.7\linewidth]{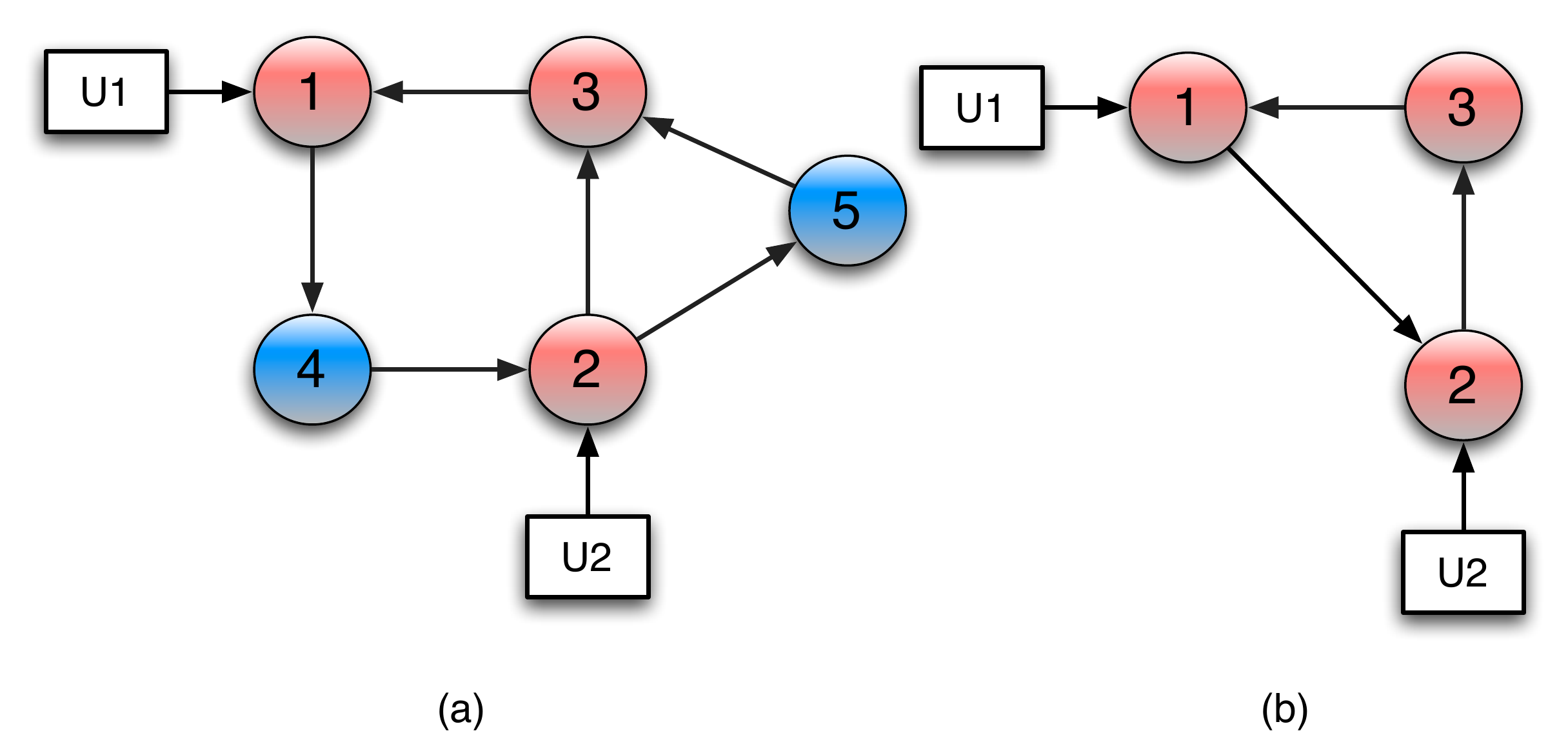}
\caption{\textbf{(a)} An example system with two inputs, three measured states (states $1$, $2$, and $3$) and two hidden states (states $4$ and $5$). \textbf{(b)} The corresponding dynamical network structure.}
\label{fig:ex1}
\end{figure}

Consistent with Proposition~\ref{th:qp}, we can check using the above that:
\begin{align*}
  \lim_{s\rightarrow \infty} s {Q}(s)&=\begin{pmatrix} 0 & 0 & a_{13}  \\
	0 & 0 & 0  \\
	0 & a_{32} & 0 \end{pmatrix} ;\\
\lim_{s\rightarrow \infty} s {P}(s)&=s\begin{pmatrix} \frac{b_{11}}{s-a_{11}} & 0 \\ 0 & \frac{b_{22}}{s-a_{22}} \\ 0 & 0 \end{pmatrix}=\begin{pmatrix}b_{11} & 0 \\ 0 & b_{22} \\ 0 & 0  \end{pmatrix}.
\end{align*}
\end{example}

Generally, there exist many realizations consistent with $[{Q},{P}]$.
In the following section, we focus on finding a $[{Q},{P}]$ minimal
realization $\left(A,B,\begin{bmatrix}
I & 0 \end{bmatrix}\right)$, i.e., a realization which is consistent with $[{Q},{P}]$
and which has minimal order, that is, with minimal dimension for $A$ (and hence the lowest possible complexity).

\section{Algorithm to find a $[{Q},{P}]$ minimal
  realization}\label{sec:optimalD}
From a dynamical structure function $[{Q},{P}]$ we cannot
reconstruct $[{W}^o, {V}^o]$ since there is no information
regarding the diagonal transfer function matrix ${R}^o$. 
Given $[{Q},{P}]$ and a diagonal proper transfer function
matrix ${R}$, a minimal realization of
$[{W}~~{V}]= [(s{I} - {R}){Q} +
{R}~~(s{I}-{R}){P}]$ can be obtained
as follows. We know that:
\begin{equation}\label{eq::WVrealization}
  [{W}~~{V}]=[{A}_{11}~~{B}_1]+{A}_{12}(s{I}-{A}_{22})^{-1}[{A}_{21}~~{B}_2]
\end{equation}
The idea is to start with an arbitrarily chosen ${R}$, and then use a state-space realization approach to find a ${R}^*$ 
which minimizes the order of a minimal realization of
$[{W}~~{V}]$. 

%\begin{lemma}\label{th:obserable} Given a dynamical structure
%  function $[{Q},{P}]$ and a diagonal proper transfer matrix ${R}$, the realization $({A},{B})$
%  obtained from eq.~\eqref{eq::WVrealization} is consistent with
%  $[{Q},{P}]$ and the pair
%  $({A}, \begin{bmatrix}{I}_{p} & {0} \end{bmatrix})$
%  is observable.
%\end{lemma}
%\begin{proof}
%  The consistency of the realization with $[{Q},{P}]$
%  follows from the definition of $[{Q},{P}]$.  From the
%  Popov-Belevitch-Hautus (PBH) rank test~\cite{zdg}, a matrix pair
%  $({A}\in \mathbb{R}^{l \times l},{C})$ is observable iff
%  \begin{equation}\label{eq:pbh}
%    \text{rank}\begin{bmatrix} s{I}-{A}\\
%      {C}\end{bmatrix}=l,
%  \end{equation}
%  for all $s\in\mathbb{C}$.  A minimal realization of
%  $\begin{bmatrix}{W} & {V}\end{bmatrix}$ implies that the
%  pair $({A}_{22},{A}_{12})$ is observable,
%  i.e.,
%  \begin{equation*} \text{rank} \begin{bmatrix}
%      s{I}_{l-p}-{A}_{22}\\
%      {A}_{12} \end{bmatrix}=l-p,~\forall s.
%  \end{equation*}
%  Hence
%  \begin{equation*} \text{rank}\begin{bmatrix}
%      s{I}-{A}_{11}&-{A}_{12}\\ -{A}_{21} &
%      s{I}_{l-p}-{A}_{22}\\ {I}_p & {0}_{p \times
%        (l-p)} \end{bmatrix}=l,~\forall s,
%  \end{equation*}
%  which concludes the proof. \end{proof}
%
%\begin{remark}
%  Given matrices ${A}$ and ${B}$ obtained in eq.~\eqref{eq::WVrealization}, the
%  dimension of ${A}$ is equal to the dimension of a minimal
%  realization of ${G}$ iff the pair $({A},{B})$ is
%  controllable.
%\end{remark}

\begin{lemma}\label{lemma:systemzero}
Suppose ${W}$, ${V}$, ${A}=\begin{bmatrix} A_{11} & A_{12} \\ A_{21} & A_{22} \end{bmatrix}$ and ${B}=\begin{bmatrix} B_{1} \\ B_{2} \end{bmatrix}$ are defined as in
eq.~\eqref{eq::WVrealization}, then ${V}$ and ${G}$ share the
same zeros.
\end{lemma}
\begin{proof}
Since $s{I}-{W}$ is the Schur complement of
  $s{I}-{A}_{22}$ in $s{I}-{A}$, then
\begin{equation}\label{eq:sI-W}
\det(s{I}-{W})=\frac{\det(s{I}-{A})}{\det(s{I}-{A}_{22})}.
\end{equation}
Recall that ${V}={B}_1+{A}_{12}(s{I}-{A}_{22})^{-1}{B}_2$. Since $(s{I}-{W}){G}={V}$, we thus have that ${V}$ and ${G}$ share the same zeros \cite[page~$153$]{Fallside}.
\end{proof}
Given a dynamical structure function $[{Q},{P}]$, a random
choice of a proper diagonal transfer function matrix ${R}$ is
likely to result in additional zeros in
${V}=(s{I}-{R}){P}$.  From
Lemma~\ref{lemma:systemzero}, this will lead to additional zeros in
${G}$ which are associated to uncontrollable eigenvalues of the
considered realization \cite[Section 4]{MIT} and of course does not lead to 
a minimal realization in eq.~\eqref{eq::WVrealization}. 
At this stage the following question arises: how can we find a proper diagonal transfer function matrix ${R}^*$ such that a minimal realization of
$[{W}~~{V}]$ is a $[{Q},{P}]$ minimal realization, i.e., 
\begin{equation}\label{eq:D}
R^*=\textbf{argmin}_R ~\text{deg}\left\{   (s{I}-{R})s^{-1}[s{Q}~~s{P}]+[{R}~~{0}]\right\},
\end{equation}
where deg is the McMillan degree \cite{zdg}. 
Note that, since there are many choices for ${R}^*$ that minimize
the order of minimal realizations of $[{W}~~{V}]$, a chosen
${R}^*$ may be different from ${R}^o$. 

Assume that all elements in $[{Q}~~{P}]$ only have simple poles. This assumption can be relaxed but we adopt it here for simplicity. Also assume that $[{Q}~~{P}]$ does not possess any poles at $0$ (otherwise we can change eq.~\eqref{eq:D} to $(s{I}-{R})(s-a)^{-1}[(s-a){Q}~~(s-a){P}]+[{R}~~{0}]$, where $a\in\mathbb{R}$ is not a pole of $[{Q}~~{P}]$).
\begin{proposition}\label{prop:d}
  Assume $[{I-Q}~~{P}]$ only has simple poles and does not have any zeros\footnote{These assumptions can be relaxed, see Section $3.6$ of \cite{yephd} for more details.}. A minimal order realization of $[{W}~~{V}]$ in~\eqref{eq::WVrealization} can be achieved using a constant diagonal matrix
  ${R}^*$.
\end{proposition}
\begin{proof}
  Assume ${R}^*$ has at least one term on the diagonal with the degree
  of the numerator greater or equal to $1$, e.g., suppose the $i^{th}$ term in
  $(s{I}-{R}^*)s^{-1}=\frac{(s+b)\epsilon_i(s)}{s\phi_i(s)}$ with any
  $b\in\mathbb{R}$ and $deg(\epsilon_i(s))=deg(\phi_i(s))\ge1$, where
  $deg(\cdot)$ returns the degree of a polynomial. 
  Hence, the product $(s{I}-{R}^*)s^{-1}[s{Q}~~s{P}]$
  will introduce $deg(\phi_i(s))$ new poles and, due the assumption of
  simple poles, can at most eliminate $deg(\epsilon_i(s))=deg(\phi_i(s))$
  poles since $[{I-Q}~~{P}]$ does not have any zeros. As a consequence, we can change the $i^{th}$ term from $\frac{(s+b)\epsilon_i(s)}{s\phi_i(s)}$ to $\frac{s+a}{s}$ without increasing the order. Doing this along all the elements of ${R}^*$ proves the result.
\end{proof}

If ${R}^*$ is a constant matrix, the term $[{R}^*~~{0}]$ in
eq.~\eqref{eq:D} is also a constant matrix. Therefore, the order of a
minimal realization is only determined by
$(s{I}-{R}^*)s^{-1}[s{Q}~~s{P}]\triangleq
{N}[s{Q}~~s{P}]$. Thus, finding the ``optimal'' ${R}^*$ which leads to the minimal order in eq.~\eqref{eq:D}
is equivalent to finding a diagonal proper transfer matrix ${N}$ \footnote{${N}$ with corresponding minimal realization
$({A}_{2},{B}_{2},{C}_{2},{I})$ is restricted to the set of matrices of the form $(s{I}-{R}^*) s^{-1}$ with a constant ${R}^*$ from Proposition~\ref{prop:d}.} such that ${N}[s{Q}~~s{P}]$ has as few poles as possible. Based on this idea, the following algorithm is proposed:

\noindent{\em {\bf Step 1:} Find a Gilbert's realization of the dynamical structure
  function.}\\
First, using the results in \cite[Lemma~1]{08net_rec}, we find a minimal realization
$({A}_{1},{B}_{1},{C}_{1},{D}_{1})$ of
$[s{Q}~~s{P}]$. When $[s{Q}~~s{P}]$ has $l$ simple
poles, using Gilbert's realization \cite{gilbert} gives
\begin{equation*}
[s{Q}~~s{P}]=\sum_{i=1}^l
\frac{{K}_i}{s-\lambda_i}+\lim_{s\rightarrow\infty}[s{Q}~~s{P}],
\end{equation*}
where ${K}_i =
\lim_{s\rightarrow\lambda_i}(s-\lambda_i)[s{Q}~~s{P}]$ and has
rank $1$ since we are assuming that $[s{Q}~~s{P}]$ has simple poles.

Consider a matrix decomposition of ${K}_i$ of the
following form:
\begin{equation*}
{K}_i={E}_i{F}_i,~\forall i,
\end{equation*}
where ${E}_i \in \mathbb{R}^{p}$ and
${F}_i=({E}_i^T{E}_i)^{-1}{E}_i^T{K}_i$.
Then ${A}_1=\text{diag}\{\lambda_i\}\in\mathbb{R}^{l\times l}$,
${B}_1=\begin{bmatrix} {F}^T_1 & {F}^T_2 & \ldots &
  {F}^T_l \end{bmatrix}^T$, ${C}_1=\begin{bmatrix}
  {E}_1 & {E}_2 & \ldots & {E}_l \end{bmatrix}$ and
${D}_1=\lim_{s\rightarrow\infty}[s{Q}~~s{P}]$.

\noindent{\em {\bf Step 2:} Find the maximal number of cancelled poles.}\\
%Define a function $f(\cdot):\mathbb{R}^{l} \to
%\mathbb{N}^{l}$ as follows: given a vector
%${E}\in\mathbb{R}^{l}$, the $i^{th}$ element of
%${y}=f({E}) \in\mathbb{N}^{l}$ is $0$ if the corresponding element in ${E}$ is zero, or otherwise is $1$.  Given
%${\alpha} \in \mathbb{R}^{n}$ and ${K}\in\mathbb{R}^{n\times
%  m}$, we will use the notation ${\alpha}[i]$ to denote its
%$i^{th}$ element and ${K}[i,:] \in \mathbb{R}^{1 \times m}$ to
%denote its $i^{th}$ row.
We define $\Phi$ as a largest subset of $\{{E}_1,\cdots,{E}_l\}$ such that all the elements in $\Phi$ are mutually orthogonal. We also define $\phi$ as the cardinality of $\Phi$. Computationally, $\phi$ can be obtained using the algorithm presented in the Appendix. We claim that $\phi$ is equal to the maximum number of poles we can eliminate (the proof is in the Appendix). Therefore, the minimal order of $[{W}~~{V}]$ is
$$l-\phi.$$
As a consequence, the order of the minimal reconstruction is the dimension of
${A}_{11}$ (the constant $p$) plus the minimal dimension of
${A}_{22}$ (obtained above): $p+l-\phi$.

\noindent{\em {\bf Step 3:} Construct ${R}^*$ to obtain the minimal reconstruction.}\\
Once we have $\Phi$, using eq.~\eqref{eq:req} and ${R}^*=s{I}-s{N}$, we know that ${N}(\lambda_i)[j,j]=0$ implies ${R}^*[j,j]=\lambda_i$. Consequently, each element in the set $\Phi$ will determine at least one element in ${R}^*$. This last fact can be used to construct ${R}^*$ element by element. Once ${R}^*$ is found, we can obtain ${A}$ and ${B}$ using eq.~\eqref{eq::WVrealization}.

%\todo{I do not think this is correct. Take $[Q \ P] = 1/(s+1) (0,1,1;
%  1,0,1)$.  In this case, there is only one $E=(1;1)$ and so
%  according to the result $\phi=1$ and we can cancel one pole. In
%  fact, with $D=diag(1,1)$ you can cancel 2 poles. It seems that this condition is only sufficient, not necessary.
%  Basically, you have proven that $\phi \leq \phi^*$ but not yet that
%  $\phi \geq \phi^*$ or that equality can be achieved, so that $\phi =
%  \phi^*$.}
%\todo{We need to add another step on how to construct $D$.}

\section{Illustrative example}\label{Sec:simu}
\begin{example}\label{ex:min_realization}
Consider a dynamical structure function $[{Q},{P}]$:
\begin{equation*}
 [{Q} \ | \ {P}]=\begin{bmatrix}
0 & \frac{1}{s+2} & \frac{1}{s+3} & |&\frac{1}{s+4}\\
\frac{1}{s+1} & 0 & \frac{1}{s+3} & |&\frac{1}{s+4} \\
\frac{1}{s+1} & \frac{1}{s+2} & 0 & |&\frac{1}{s+4}
 \end{bmatrix}.
\end{equation*}
We first compute the McMillan degree of the corresponding transfer function: $\text{deg}\{G\}=\text{deg}\{(I-Q)^{-1}P)\}=4$, meaning that a $4^{th}$ order state-space model is enough to realize the transfer function. It is interesting to look at the minimal order realization consistent with the dynamical structure function. The different steps of the algorithm proposed in the previous section successively yield the following:

{\em Step 1:}
A minimal Gilbert realization of $s[{Q},{P}]$ is
\begin{align*}
&{A}_1=\text{diag}\{-1,-2,-3,-4\}, ~ {B}_1=\text{diag}\{2,2,2,4\},\\
&{C}_1=\begin{bmatrix}
 0  &  -1 & -1.5  &  -1\\
     -0.5  &   0 & -1.5   & -1\\
   -0.5 &   -1    & 0 &   -1
\end{bmatrix},~
{D}_1=\begin{bmatrix}
      0 &  1  & 1  & 1\\
      1 &  0  & 1 &  1\\
   1  & 1 &  0 &  1
\end{bmatrix}.
\end{align*}

{\em Step 2:}
By definition, ${E}_i={C}_1{v}_i$ where
${v}_i \in \mathbb{R}^{4}$ has $1$ in its $i^{th}$ position and
zero otherwise. Thus,
$$\{{E}_1,\cdots,{E}_4\}=
\left\{\begin{bmatrix}0\\-0.5\\-0.5\end{bmatrix},
  \begin{bmatrix}-1\\0\\-1\end{bmatrix}, \begin{bmatrix}-1.5\\-1.5\\0 \end{bmatrix},\begin{bmatrix}-1\\-1\\-1\end{bmatrix}\right\}.$$
Furthermore, $\phi$ is $1$ and the order of a minimal realization of the given dynamical structure function is $p+l-\phi=3+4-1=6$. Hence,
the system must contain at least $3$ hidden states.

{\em Step 3:} ${R}^*$ can be chosen as
$\text{diag}\{a,-1,-1\}$, $\text{diag}\{-2,a,-2\}$,
$\text{diag}\{-3,-3,a\}$, or $\text{diag}\{-4,-4,-4\}$ for any
$a\in\mathbb{R}$.

The reconstructed networks are represented in Fig.~\ref{Fig}. There are three measured (red) nodes, labeled $1,2,3$ and by the analysis above, there are at least three hidden nodes such that the corresponding realization is consistent with the dynamical structure function. The red connections between measured nodes are the same for all candidate networks which is in accordance with Proposition~\ref{th:qp}. Dashed lines correspond to the connections between hidden and measured nodes. 

\begin{figure}[ht]
\centering
\includegraphics[scale=0.3]{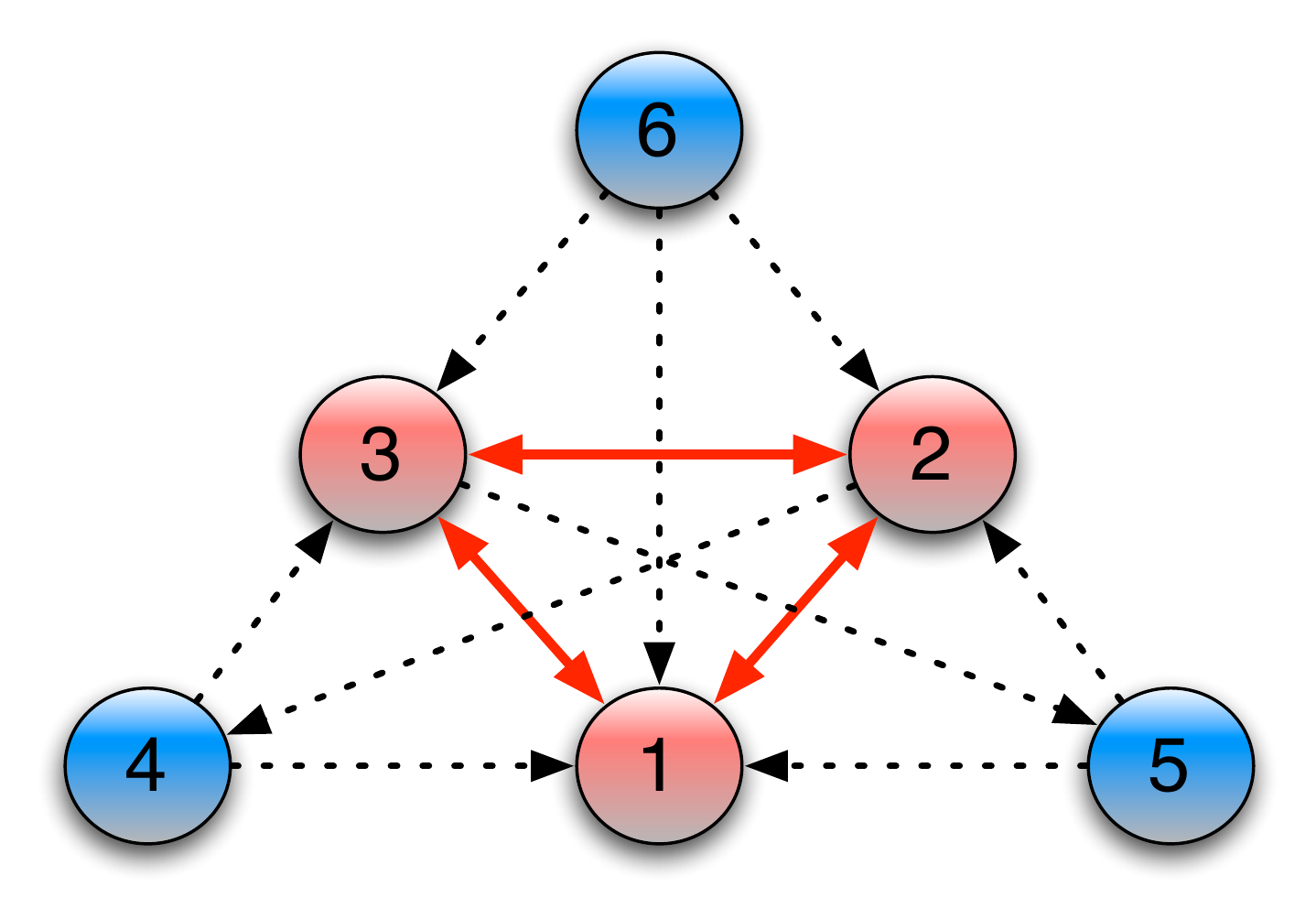}
\includegraphics[scale=0.3]{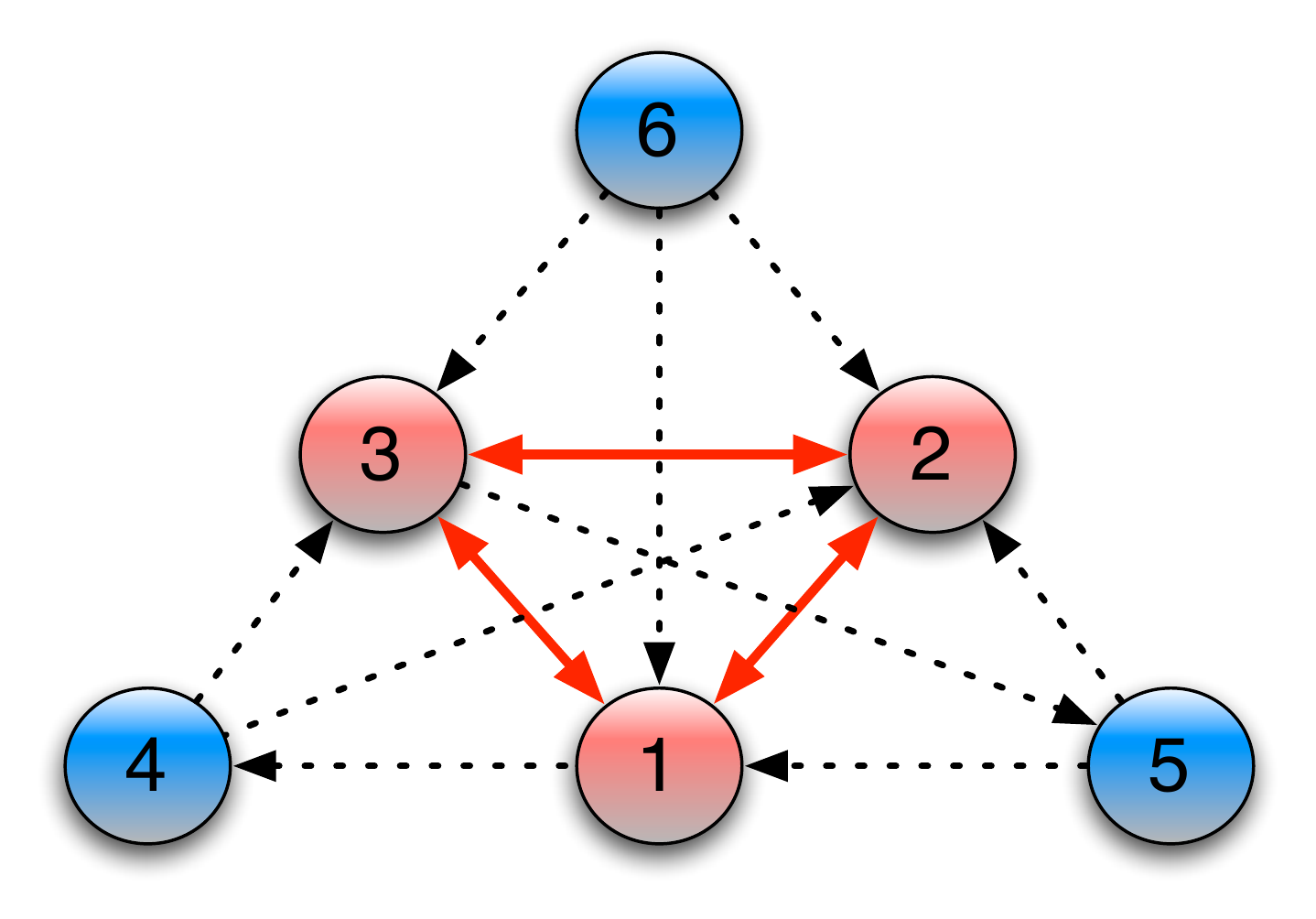}
\includegraphics[scale=0.3]{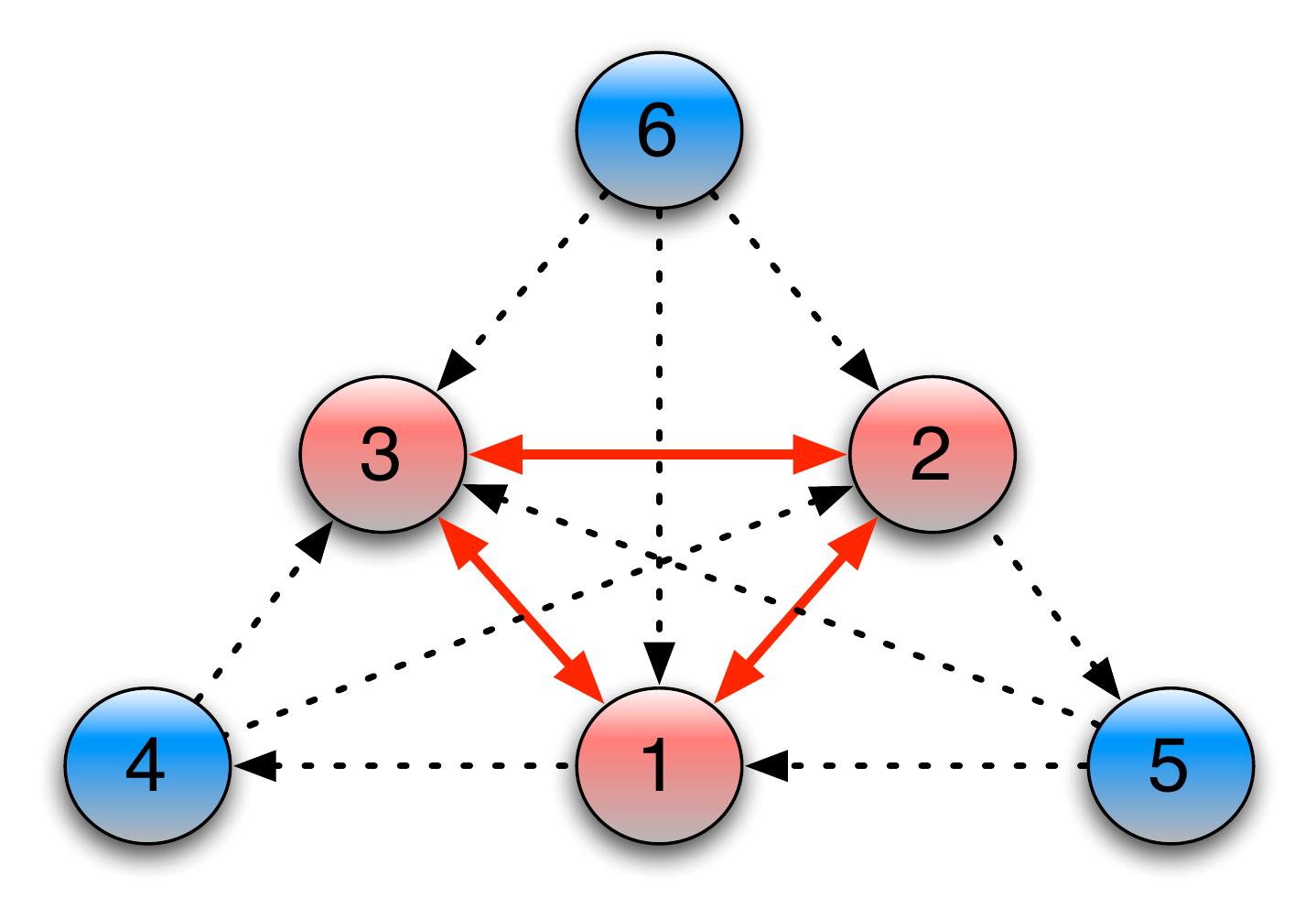}
\includegraphics[scale=0.3]{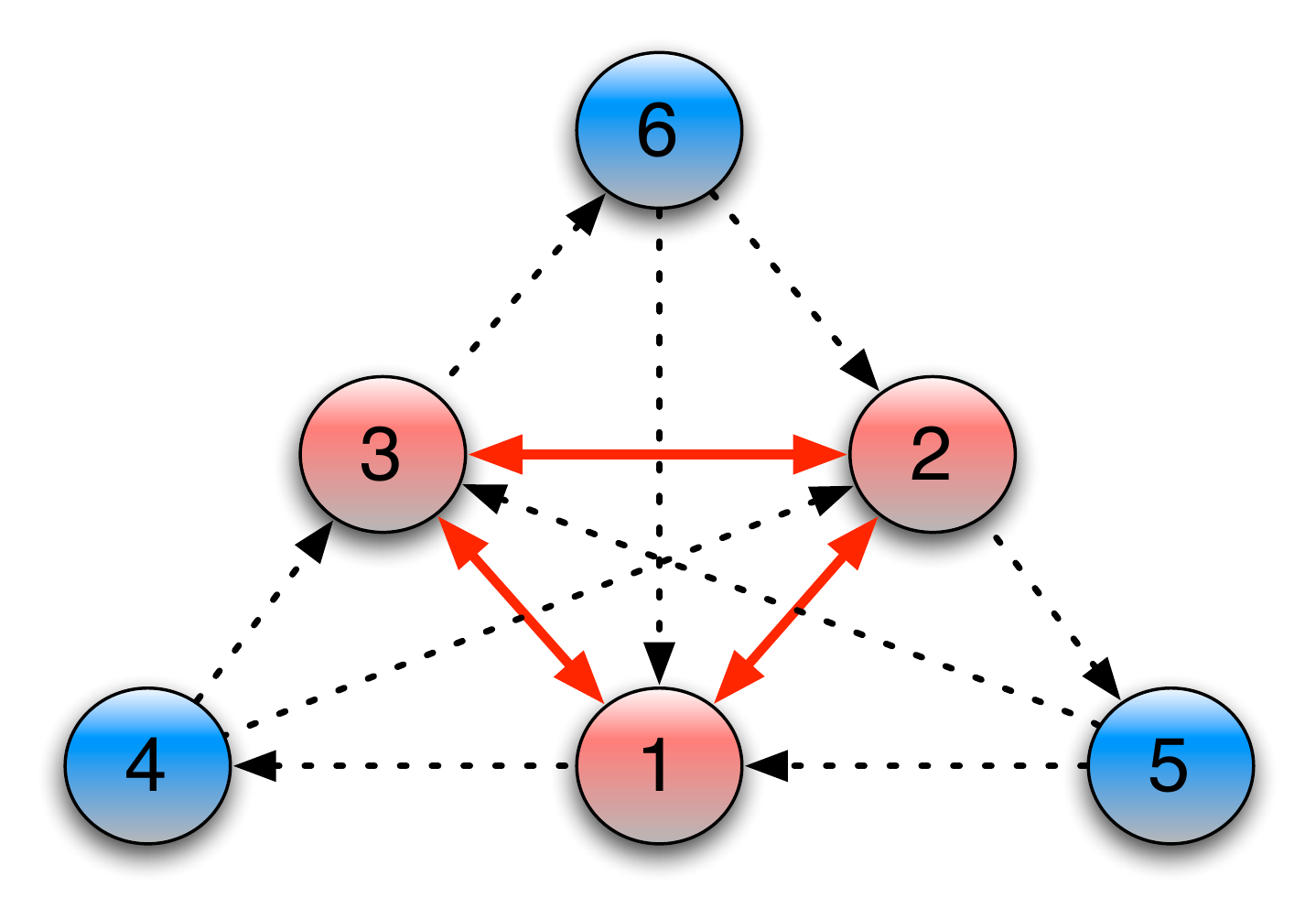}
\caption{Topologies corresponding to the four $[{Q},{P}]$ minimal realizations. 
The measured nodes are colored red, while the hidden ones blue. Red connections between
measured nodes are the same for all the networks due to Proposition~\ref{th:qp}. Each node has a self-loop but we omit it for simplicity.}
\label{Fig} 
\end{figure}

From a biological perspective, this indicates that there are at least $3$ unmeasured species interacting with the measured species. Of course, the ``true'' biological system might be even more complicated, i.e., it might have more than $6$ species. Yet, when more states are measured, the dynamical structure functions can be easily updated and a new search for a minimal realization of the updated system can be performed to reveal the corresponding minimal number of hidden states.
\end{example}

\section{Conclusions}\label{sec:conclusion}

In this paper, we have presented a method for obtaining a minimal order realization consistent with a given dynamical structure function. We show that the minimal order realization of a given dynamical structure function can be achieved by choosing a constant diagonal matrix ${R}^*$. This provides a way to estimate the complexity of the
system by determining the minimal number of hidden states that needs to be considered in the reconstructed network.  For
example, in the context of reconstruction of biological networks from data, it helps to understand the minimal number of
unmeasured molecules in a particular pathway.

%The paper explored an interesting approach to find minimal realization for a dynamical structure function. Properties of the
%dynamical structure functions were investigated and based on them,
%heuristic algorithm was implemented as tool to find such minimal realizations with low complexity and therefore easy to be implemented.

\section{Acknowledgement}
The authors want to  thank the anonymous reviewers for their feedback which has helped in improving the quality of the manuscript. Guy-Bart Stan, Ye Yuan and Jorge Gon\c{c}alves gratefully acknowledge the support of EPSRC under the projects EP/E02761X/1, EP/I03210X/1 and the support of Microsoft Research through the PhD Scholarship Program of Ye Yuan.  Sean Warnick gratefully acknowledges the support of grants AFRL FA8750-09-2-0219 and FA8750-11-1-0236.  Guy-Bart Stan also gratefully acknowledges the support of the EPSRC Centre for Synthetic Biology and Innovation at Imperial College, London.

\appendix
\noindent\textbf{Proof of the claim in Step $2$ of the proposed algorithm:}
\begin{proof}
Using results from Section $4$ of \cite{MIT}, if a
pole of $[s{Q}~~s{P}]$, say $\lambda_i$, is cancelled by
${N}=(s{I}-{R}^*)s^{-1}\triangleq {C}_2({A}_2-sI)^{-1}{B}_2+{I}$, then the realization of the
cascade $(s{I}-{R})s^{-1}[s{Q}~~s{P}]$ loses
observability. 
In this case, it follows that there exists a nonzero vector
${w}_i= [{w}_{1,i}^T, {w}_{2,i}^T]^T$ such that
$$\begin{bmatrix} {A}_1-\lambda_i{I} & {0} \\ {B}_2
  {C}_1 & {A}_2-\lambda_i{I} \\ {C}_1& {C}_2
\end{bmatrix}\begin{bmatrix} {w}_{1,i}\\ {w}_{2,i} \end{bmatrix}={0}.$$
The first equation shows that ${w}_{1,i}$ is an eigenvector of
${A}_1$ corresponding to $\lambda_i$. Since
${A}_1$ is diagonal, ${w}^T_{1,i}=\begin{bmatrix}0&\ldots&0&1_{i^{th}}&0&\ldots&0\end{bmatrix} \in
\mathbb{R}^{1 \times l}$. Therefore, we have
$$\begin{bmatrix} {A}_2-\lambda_i{I} & {B}_2\\ {C}_2& {I}
\end{bmatrix}\begin{bmatrix} {w}_{2,i} \\
  {C}_1{w}_{1,i} \end{bmatrix}={0}.$$
Noticing that ${C}_1{w}_{1,i}= {E}_{i}$ and that
\begin{align*}
&\begin{bmatrix} {I} & {0} \\
-{C}_2({A}_2-s{I})^{-1}& {I}
\end{bmatrix}\begin{bmatrix} {A}_2-s{I} & {B}_2\\ {C}_2& {I}
\end{bmatrix}
=\begin{bmatrix} {A}_2-s{I} & {B}_2\\ {0} &
{N(s)}
\end{bmatrix},
\end{align*}
we obtain, since $\lambda_i\neq0$ is not a pole of $N$,
\begin{equation}\label{eq:req}
{N}(\lambda_i){E}_i={0}.
\end{equation}
In summary, designing ${R}^*$ to cancel any pole
$\lambda_i$ of $[s{Q}~~s{P}]$ is equivalent to imposing that
eq.~\eqref{eq:req} holds. The next question is: given
$[s{Q}~~s{P}]$ what is the maximal number of poles that can
be cancelled by ${N}$, i.e., what is the largest number of poles for which
eq.~\eqref{eq:req} is satisfied?

To answer this, notice that ${E}_i[j]$ being nonzero for some $j$, implies
that there exists at least one nonzero element in the $j^{th}$ row
of ${E}_i$. In this case, satisfying eq.~\eqref{eq:req}
imposes that the $j^{th}$ diagonal element of
${N}(\lambda_i)$ is $0$, i.e., the $j^{th}$
diagonal element of ${R}^*$ is $\lambda_i$. In other words, a
nonzero element in ${E}_i$ corresponds to a fixed value in
the corresponding diagonal position in ${R}^*$. Since ${R}^*$ is a constant diagonal matrix then any pair of
orthogonal vectors in $\{{E}_1,\cdots,{E}_l\}$ does
not intervene in the choice of an element on the diagonal of
${R}^*$.
\end{proof}

\noindent\textbf{Algorithm to find $\phi$ and $\Phi$:}\\
As is presented in \cite{GA}, an undirected graph is
denoted by $\mathcal{G}=(\mathcal{V},\mathcal{E})$
where $\mathcal{V}=\left\{\nu_{1},\ldots,\nu_{l}\right\}$ is the set
of nodes and $\mathcal{E}\subset\mathcal{V}\times\mathcal{V}$ is the
set of edges.

For our purposes, we construct an undirected graph $\mathcal{G}_a$ using the following rules:
\begin{itemize}
\item A node is associated with each vector in the set $\{{E}_1,\cdots,{E}_l\}$. There are thus $l$ nodes in the considered graph.
\item An undirected edge $(i,j)$ is drawn between node $i$ and node $j$ if the equality ${E}_i^T{E}_j=0$ is satisfied.
\end{itemize}

It is easy to see that the maximum cardinality of the set
$\Phi$ corresponds to the maximum number of nodes in a complete
subgraph $K_n$ of the graph $\mathcal{G}_a$. 

Although the problem of finding a largest complete subgraph in an undirected graph is a NP-hard problem,
methods to this end have been well-studied in \cite{link}.\footnote{Some corresponding MATLAB code can be downloaded from http://www.mathworks.com/matlabcentral/fileexchange/19889.} To our best knowledge, for an arbitrary graph, the fastest algorithm has a complexity of $\mathcal{O}(2^{n/4})$\cite{rob}. Therefore, we can use these methods to
obtain a largest complete subgraph and consequently compute the corresponding set $\Phi$ and its corresponding cardinality $\phi$.
\end{document}